\documentclass[11pt]{article}

\title{Prophet Secretary\thanks{A preliminary version of this paper is to appear in Proceedings of the 23rd European Symposium on Algorithms (ESA 2015), Patras, Greece.}}

\author{ 
Hossein Esfandiari\thanks{University of Maryland, College Park, MD, USA. Supported in part by NSF CAREER award 1053605, NSF grant CCF-1161626, ONR YIP award N000141110662, DARPA/AFOSR grant FA9550-12-1-0423. {\tt hossein@cs.umd.edu}.}
\and
MohammadTaghi Hajiaghayi\thanks{University of Maryland, College Park, MD, USA. Supported in part by NSF CAREER award 1053605, NSF grant CCF-1161626, ONR YIP award N000141110662, DARPA/AFOSR grant FA9550-12-1-0423. {\tt Hajiagha@cs.umd.edu}.}
\and
Vahid Liaghat \thanks{Management Science and Engineering and  
Institute for Computational and Mathematical Engineering, 
Stanford University. Supported in part by a Google PhD Fellowship. {\tt vliaghat@stanford.edu}.}
\and
Morteza Monemizadeh\thanks{Computer Science Institute of Charles University,
Faculty of Mathematics and Physics, Prague, Czech Republic. Partially supported by the project 14-10003S of GA \v{C}R. Part of this work was done when the author was at Department of Computer Science, Goethe-Universit\"{a}t Frankfurt, Germany and supported in part by MO 2200/1-1. {\tt monemi@iuuk.mff.cuni.cz}.}
}

\usepackage{makeidx}
\usepackage{euscript}
\usepackage{amsmath,amssymb,amsfonts}
\usepackage{dsfont}
\usepackage{latexsym}
\usepackage{subfigure}
\usepackage{graphicx}
\usepackage{times}\usepackage[scaled=0.92]{helvet} 
\usepackage{euler}
\usepackage{fancybox}
\usepackage{fullpage}
\usepackage{hyperref}
\usepackage{paralist}

\usepackage{color}
\usepackage{wrapfig}
\usepackage{tikz}

\usepackage[small,compact]{titlesec}

\newcount\shortyear\newcount\shorthour\newcount\shortminute
\shorthour=\time\divide\shorthour by 60\shortyear=\shorthour
\multiply\shortyear by 60\shortminute=\time\advance\shortminute by
-\shortyear \shortyear=\year\advance\shortyear by -1900

\def\zeit{\number\shorthour:\ifnum\shortminute<10 0\number\shortminute
\else\number\shortminute\fi}

\newenvironment{proof}{\noindent{\bf Proof : \ }}{\hfill$\Box$\par\medskip}

\newtheorem{theorem}{Theorem}
\newtheorem{corollary}[theorem]{Corollary}
\newtheorem{lemma}[theorem]{Lemma}
\newtheorem{proposition}[theorem]{Proposition}

\newtheorem{remark}[theorem]{Remark}

\newenvironment{proofof}[1]{\begin{trivlist} \item {\bf Proof
#1:~~}}
  {\qed\end{trivlist}}

\newcommand{\COMMENTED}[1]{{}}

\newcommand{\OPT}{OPT}

\newcommand{\prob}[1]{\operatorname{Pr}\left[#1\right]}
\newcommand{\ex}[1]{\operatorname{E}\left[#1\right]}

\newcommand{\R}{\ensuremath{\mathbb{R}}}

\newcommand{\eat}[1]{}

\newlength{\savedparindent}
\newcommand{\SaveIndent}{\setlength{\savedparindent}{\parindent}}
\newcommand{\RestoreIndent}{\setlength{\parindent}{\savedparindent}}

\newcommand{\InGrayMiddle}[1]{%
\SaveIndent{} %
\centerline{ \fcolorbox[rgb]{0,0,0}{0.95,0.95,0.95}{
\begin{minipage}{0.8\linewidth} %
\RestoreIndent{}%
#1
\end{minipage}
} } }

\begin{document}

\begin{titlepage}
\maketitle\thispagestyle{empty}


\begin{abstract}

Optimal stopping theory is a powerful tool for analyzing scenarios 
such as online auctions in which we generally require optimizing 
an objective function over the
space of stopping rules for an allocation process under uncertainty. 
Perhaps the most classic problems of stopping theory are the prophet inequality problem
and the secretary problem. 
The classical prophet inequality states that by choosing the same threshold $\OPT/2$ for every step, 
one can achieve the tight competitive ratio of $0.5$. 
On the other hand, for the basic secretary problem, 
the optimal strategy achieves the tight competitive ratio of $1/e\approx 0.36$

In this paper,
we introduce \emph{prophet secretary}, a natural combination of the prophet inequality 
and the secretary problems. 
An example motivation for our problem is as follows. 
Consider a seller that has an item to sell on the market to a set of arriving customers. 
The seller knows the types of customers that may be interested in the item and 
he has a price distribution for each type: the price offered by a customer of a type 
is anticipated to be drawn from the corresponding distribution. However, 
the customers arrive in a random order. Upon the arrival of a customer, 
the seller makes an irrevocable decision whether to sell the item at the offered price. 
We address the question of finding a strategy for selling the item at a high price.
In particular, 
we show that by using a single uniform threshold one cannot break the $0.5$ barrier of the prophet inequality for the prophet secretary problem. However, we show that
\begin{itemize}
\item using $n$ distinct non-adaptive thresholds one can obtain a competitive ratio that goes to $(1-1/e \approx 0.63)$ as $n$ grows; and 
\item no online algorithm can achieve a competitive ratio better than $0.75$.
\end{itemize}
Our results improve the (asymptotic) approximation guarantee of single-item sequential posted pricing mechanisms from $0.5$ to $(1-1/e)$ when the order of agents (customers) is chosen randomly.

We also consider the minimization variants of stopping theory problems and in particular the prophet secretary problem. 
Interestingly, we show that, even for the simple case in which 
the input elements are drawn from identical and independent distributions (i.i.d.), 
there is no constant competitive online algorithm for the minimization 
variant of the prophet secretary problems.
We extend this hardness result to the minimization variants of both the prophet inequality and the secretary problem as well.

\end{abstract}


\end{titlepage}

\section{Introduction}
Optimal stopping theory is a powerful tool for analyzing scenarios
in which we generally require optimizing an objective function over the
space of stopping rules for an allocation process under uncertainty. 
One such a scenario is the online auction which is the essence of many modern markets, particularly
networked markets where information about goods, agents, and
outcomes is revealed over a period of time and the agents must make
irrevocable decisions without knowing future information. 
Combining optimal stopping theory with game theory allows us to
model the actions of rational agents applying competing stopping
rules in an online market.

Perhaps the most classic problems of stopping theory are the \emph{prophet inequality} and the \emph{secretary problem}. 
Research investigating the relation between online auction mechanisms and prophet 
inequalities was initiated by Hajiaghayi, Kleinberg, and Sandholm \cite{HKS07}. They observed that 
algorithms used in the derivation of prophet inequalities, owing to their monotonicity properties, 
could be interpreted as truthful online auction mechanisms and that the prophet inequality 
in turn could be interpreted as the mechanism's approximation guarantee. 
Later Chawla, Hartline, Malec, and Sivan~\cite{CHMS10} showed the applications of prophet inequalities in Bayesian optimal mechansim design problems.
The connection between the secretary problem and online auction mechanisms has been explored  
by Hajiaghayi, Kleinberg and Parkes~\cite{HKP04} 
and initiated several follow-up papers (see e.g.~\cite{BIKK07,BIKK08,BIK07,IKM06,K05}).


\paragraph{Prophet Inequality.}
The classical prophet inequality has been studied in the optimal stopping theory 
since the 1970s when introduced by Krengel and Sucheston \cite{K78,KS77,KS78} and 
more recently in computer science Hajiaghayi, Kleinberg and  Sandholm \cite{HKS07}. 
In the prophet inequality setting, given (not necessarily identical) distributions $\{D_1,\ldots,D_n\}$, an online sequence of 
values $X_1 ,\cdots, X_n$ where $X_i$ is drawn from $D_i$, an onlooker has to choose one item from 
the succession of the values, where $X_k$ is revealed at step $k$.  
The onlooker can choose a value only at the time of arrival. The onlooker's goal is to maximize her revenue. 
The inequality has been interpreted as meaning that a prophet with complete 
foresight has only a bounded advantage over an onlooker who observes 
the random variables one by one, and this explains the name \emph{prophet inequality}. 

An algorithm for the prophet inequality problem can be described by setting a threshold for every step: we stop at the first step that the arriving value is higher than the threshold of that step.
The classical prophet inequality 
states that by choosing the same threshold $\OPT/2$ for every step, 
one achieves the competitive ratio of $1/2$.
Here the optimal solution $\OPT$ is defined as $\ex{\max X_i}$.
Naturally, 
the first question is whether one can beat $1/2$. Unfortunately, this is not possible: let $q = \frac{1}{\epsilon}$, and $q'= 0$. The first value $X_1$ is always $1$. 
The second value $X_2$ is either $q$ with probability $\epsilon$ or $qâ'$ with 
probability $1 -\epsilon$. Observe that the expected revenue of any (randomized) online 
algorithm is $\max(1,\epsilon(\frac{1}{\epsilon})) = 1$. 
However the prophet, i.e., the optimal offline solution
would choose $q'$ if it arrives, and he would choose the first value otherwise. 
Hence, the optimal offline revenue is $(1 - \epsilon) \times 1 + \epsilon(\frac{1}{\epsilon}) \approx 2$. 
Therefore we cannot hope to break the $1/2$ barrier using any online algorithm.


\paragraph{Secretary Problem.}
Imagine that you manage a
company, and you want to hire a secretary from a pool of $n$
applicants. You are very keen on hiring only the best and brightest.
Unfortunately, you cannot tell how good a secretary is until you
interview him, and you must make an irrevocable decision whether or
not to make an offer at the time of the interview. The problem is to
design a strategy which maximizes the probability of hiring the most
qualified secretary. It is well-known since 1963 by Dynkin in \cite{Dynkin63}
that the optimal policy is to interview the first $t-1$ applicants,
then hire the next one whose quality exceeds that of the first $t-1$
applicants, where $t$ is defined by $\sum_{j=t+1}^n \frac{1}{j-1}  
\le 1 < \sum_{j=t}^n \frac{1}{j-1}$. As $n \rightarrow \infty$, the
probability of hiring the best applicant approaches $1/e\approx 0.36$, as does
the ratio $t/n$. Note that a solution to the secretary problem
immediately yields an algorithm for a slightly different objective
function optimizing the expected value of the chosen element.
Subsequent papers have extended the problem by varying the objective
function, varying the information available to the decision-maker,
and so on, see e.g.,~\cite{ajtai01,GHB83,Vanderbei80,Wilson91}.


\subsection{Further Related Work}\label{sec:related}
\paragraph{Prophet Inequality.}
The first generalization of the basic prophet inequality introduced by Krengel 
and Sucheston \cite{K78,KS77,KS78} is \emph{the multiple-choices} 
prophet inequality \cite{Assaf01ratioprophet} in which both the onlooker and the prophet
have $k > 1$ choices. Currently, the best algorithm for this setting is due to Alaei \cite{A11}, 
who gave an online algorithm with  $(1 - \frac{1}{\sqrt{k + 3}})$-competitive ratio 
for $k$-choice optimal stopping. Besides this, we have two generalizations 
for the (multiple-choices) prophet inequality that are  
\emph{matroid} prophet inequality \cite{DBLP:conf/stoc/KleinbergW12} 
and \emph{matching} prophet inequality  \cite{DBLP:conf/sigecom/AlaeiHL12}.

In the matroid prophet inequality, 
we are given a matroid whose elements have random weights 
sampled independently from (not necessarily identical) probability 
distributions on $\R^+$. 
We then run an online algorithm with knowledge of the matroid structure 
and of the distribution of each element's weight. 
The online algorithm must then choose irrevocably an independent
subset of the matroid by observing the sampled value of each
element (in a fixed, prespecified order). 
The online algorithm's payoff is defined to be the
sum of the weights of the selected elements. 
Kleinberg and Weinberg \cite{DBLP:conf/stoc/KleinbergW12} show that 
for every matroid, there is an online algorithm
whose expected payoff is at least half of the expected weight of
the maximum-weight basis. Observe that the original prophet inequality 
introduced by Krengel and Sucheston \cite{K78,KS77,KS78} 
corresponds to the special case of rank-one matroids.

The matching prophet inequality is due to 
Alaei, Hajiaghayi, and Liaghat \cite{alaei2011adcell,DBLP:conf/sigecom/AlaeiHL12,alaei2013online}. 
Indeed, they study the problem of online prophet-inequality matching 
in bipartite graphs. There is a static set of bidders and an
online stream of items. The interest of bidders 
in items is represented by a weighted bipartite graph. Each bidder has a capacity,
i.e., an upper bound on the number of items that can be allocated to him. 
The weight of a matching is the total weight of edges
matched to the bidders. Upon the arrival of an item, the online algorithm 
should either allocate it to a bidder or discard it. 
The objective is to maximize the weight of the resulting matching. 
Here we assume we know the distribution of the incoming items in advance and 
we may assume that the $t^{th}$ item is drawn from distribution $\mathcal{D}_t$. 
They generalize the $\frac{1}{2}$-competitive ratio of Krengel and Sucheston \cite{KS77} 
by presenting an algorithm with an approximation ratio of $1-\frac{1}{\sqrt{k+3}}$ 
where $k$ is the minimum capacity. 
Oberve that the classical prophet inequality is a special case of this model 
where we have only one bidder with capacity one, i.e., $k = 1$ 
for which they get the same $\frac{1}{2}$-competitive ratio.


\paragraph{Secretary Problem.}
The first generalization of the basic secretary problem \cite{Dynkin63} is 
the \emph{multiple-choice secretary problem} \cite{BIKK08} 
(see a survey by Babaioff et al.~\cite{BIKK08})) in which the
interviewer is allowed to hire up to $k\geq 1$ applicants in order
to maximize performance of the secretarial group  based on their
overlapping skills (or the joint utility of selected items in a more
general setting). More formally, assuming applicants of a set $S=
\{a_1, a_2, \cdots, a_n \}$ (applicant pool) arriving in a uniformly
random order, the goal is to select a set of at most $k$ applicants
in order to maximize a non-negative profit function $f:2^{S}\mapsto\R^{\geq 0}$. For example, when $f(T)$
is the maximum individual value~\cite{F83,GM66}, or when $f(T)$ is
the sum of the individual values in $T$~\cite{K05}, the problem has
been considered thoroughly in the literature. 
Beside this, two generalizations for the (multiple-choices) 
secretary problem are {\em submodular secretary} \cite{DBLP:journals/talg/BateniHZ13} 
and {\em matroid secretary}  \cite{BIK07}.

The submodular secretary problem is introduced by Bateni, Hajiaghayi, and Zadimoghaddam 
\cite{DBLP:journals/talg/BateniHZ13}.   
Indeed, both of the maximum individual value~\cite{F83,GM66} and 
the sum of the individual values ~\cite{K05} aforementioned 
are special monotone non-negative submodular functions. 
Bateni, Hajiaghayi, and Zadimoghaddam \cite{DBLP:journals/talg/BateniHZ13} 
give an online algorithm with $(\frac{7}{1-1/e})$-competitive ratio for 
the submodular secretary problem. We should mention that there are more recent   
results with better constant competitive ratio (See for example 
the references in \cite{DBLP:journals/talg/BateniHZ13}).

In the matroid secretary problem considered by Babaioff et al.~\cite{BIK07},
we are given a matroid with a ground set ${\cal U}$ of
elements and a collection of independent (feasible) subsets ${\cal
	I}\subseteq 2^{\cal U}$  describing the sets of elements which can
be simultaneously accepted. The goal is to design online algorithms in which the
structure of ${\cal U}$ and ${\cal I}$ is known at the outset
(assume we have an oracle to answer whether a subset of ${\cal U}$
belongs to ${\cal I}$ or not), while the elements and their values
are revealed one at a time in a random order. As each element is
presented, the algorithm must make an irrevocable decision to select
or reject it such that the set of selected elements  belongs to
${\cal I}$ at all times. Babaioff et al.~\cite{BIK07} present an $O(\log
r)$-competitive algorithm for general matroids, where $r$ is the
rank of the matroid. 
However, they leave as a main open question the existence of constant-competitive
algorithms for general matroids. On the other hand, there are several 
follow-up works including the recent FOCS 2014 paper 
due to Lachish \cite{L14} which gives $O(\log\log \text{rank})$-competitive 
algorithm.


\section{Our Contributions}
In this paper, we introduce \emph{prophet secretary} as a natural combination of the prophet inequality problem
and the secretary problem with applications to the Bayesian optimal mechanism design.
Consider a seller that has an item to sell on the market to a set of arriving customers. 
The seller knows the types of customers that may be interested in the item and 
he has a price distribution for each type: the price offered by a customer of a type 
is anticipated to be drawn from the corresponding distribution. However, 
the customers arrive in a random order. Upon the arrival of a customer, 
the seller makes an irrevocable decision to whether sell the item at the offered price. 
We address the question of maximizing the seller's gain.

More formally, in the prophet secretary problem we are given 
a set $\{D_1,\ldots,D_n\}$ of (not necessarily  identical) distributions. 
A number $X_i$ is drawn from each distribution $D_i$ and then, 
after applying a random permutation $\pi_1,\ldots, \pi_n$, 
the numbers are given to us in an online fashion, i.e., 
at step $k$, $\pi_k$ and $X_{\pi_k}$ are revealed. 
We are allowed to choose only one number, which can be done 
only upon receiving that number. The goal is to maximize 
the expectation of the chosen value, compared to the expectation 
of the optimum offline solution that knows the drawn values 
in advance (i.e., $\OPT=\ex{\max_i X_i}$). For the ease of notation, 
in what follows the index $i$ iterates over the distributions 
while the index $k$ iterates over the arrival steps.

An algorithm for the prophet secretary problem can be described by a sequence 
of (possibly adaptive) thresholds $\langle \tau_1,\ldots, \tau_n \rangle$: 
we stop at the first step $k$ that $X_{\pi_k}\geq \tau_k$. In particular, if the thresholds are non-adaptive, 
meaning that they are decided in advance, the following is a generic description of an algorithm. 
The competitive ratio of the following algorithm is defined as $\frac{\ex{Y}}{\OPT}$.



\mbox{}

\InGrayMiddle{\mbox{}\\ \textbf{Algorithm Prophet Secretary \\}
\textbf{Input:} A set of distributions $\{ D_1,\ldots, D_n \}$; 
a randomly permuted stream of numbers $(X_{\pi_1}, \dots, X_{\pi_n})$ drawn from the corresponding distributions. \\
\textbf{Output:} A number $Y$. \\[-6mm]
\begin{enumerate}
\item Let $\langle \tau_1,\ldots, \tau_n \rangle$ be a sequence of thresholds.
\item For  $k \gets 1$ to $n$
    \begin{enumerate}
    \item If $X_{\pi_k}\geq \tau_k$ then let $Y=X_{\pi_k}$ and exit the {\sf For} loop.
    \end{enumerate}
    \item Output $Y$ as the solution.
\end{enumerate}
\mbox{}\\[-0.35in]}

\mbox{}


Recall that when the arrival order is adversarial, the classical prophet inequality 
states that by choosing the same threshold $\OPT/2$ for every step, 
one achieves the tight competitive ratio of $1/2$. 
On the other hand, for the basic secretary problem where the distributions are not known,
the optimal strategy is to let 
$\tau_1=\dots=\tau_{\frac{n}{e}}=\infty$ and $\tau_{\frac{n}{e}+1}=\dots=\tau_{n}=\max(X_{\pi_1}, \dots, X_{\pi_{\frac{n}{e}}})$. This leads to the optimal competitive ratio of $\frac{1}{e}\simeq 0.36$.
Hence, our goal in the prophet secretary problem is 
to beat the $1/2$ barrier.

We would like to mention that in an extension of this problem, in which the seller has $B$ identical items to sell, it is indeed easier to track the optimal solution since we have multiple choices. In fact, an algorithm similar to that of \cite{DBLP:conf/sigecom/AlaeiHL12,alaei2013online} can guarantee a competitive ratio of $1-\frac{1}{\sqrt{B+3}}$ which goes to one as $B$ grows.

We first show that unlike the prophet inequality setting, 
one cannot obtain the optimal competitive ratio by using a single uniform threshold. 
Indeed, in Section \ref{sec:1:threshold} we show  
that $1/2$ is the best competitive ratio one can achieve with uniform thresholds.	
To beat the $\frac{1}{2}$ barrier, as a warm up we first show in Section \ref{sec:2thresholds} that by using two thresholds 
one can achieve the competitive ratio of $5/9\simeq 0.55$. 
This can be achieved by choosing the threshold $\frac{5}{9}\cdot \OPT$ for the first half of 
the steps and then decreasing the threshold to $\frac{\OPT}{3}$ for the second half of the steps. 
Later in Section~\ref{sec:generaln}, 
we show that by setting $n$ distinct thresholds one can obtain the $(1-1/e \approx 0.63)$-competitive ratio 
for the prophet secretary problem.

\begin{theorem}
\label{thm:n:thresholds}
Let $\langle \tau_1,\ldots, \tau_n \rangle$ be a non-increasing sequence of $n$ thresholds, 
such that (i) $\tau_k=\alpha_k \cdot OPT$ for every $k\in [n]$; (ii) $\alpha_{n}=\frac{1}{n+1}$; 
and (iii) $\alpha_k=\frac{n\alpha_{k+1}+1}{n+1}$ for $k\in [n-1]$. 
The competitive ratio of Algorithm {\sf Prophet Secretary} invoked with thresholds $\tau_k$'s 
is at least $\alpha_1$. When $n$ goes to infinity, $\alpha_1$  converges to 
$1-1/e\approx 0.63$.
\end{theorem}

\begin{remark}
We should mention that Yan in \cite{DBLP:conf/soda/Yan11} 
establishes a 1 - 1/e approximation when the designer is allowed 
to choose the order of arrival. Thus, Theorem \ref{thm:n:thresholds} 
can be viewed as improving that result by showing that a random arrival order is sufficient to obtain the same approximation.
\end{remark}

The crux of the analysis of our algorithm is to compute the probability of picking a value $x$ 
at a step of the algorithm with respect to the threshold factors $\alpha_k$'s. 
Indeed one source of difficulty arises from the fundamental dependency between the steps: 
for any step $k$, the fact that the algorithm has not stopped in the previous steps leads to 
various restrictions on what we expect to see at the step $k$. For example, 
consider the scenario that $D_1$ is $1$ with probability one and $D_2$ is either 
$2$ or $0$ with equal probabilities. Now if the algorithm chooses $\tau_1=1$, 
then it would never happen that the algorithm reaches step two and receives a number 
drawn from $D_2$! That would mean we have received a value from $D_1$ at the first step 
which is a contradiction since we would have picked that number. In fact, the optimal strategy 
for this example is to shoot for $D_2$! We set $\tau_1=2$ so that we can ignore the first value 
in the event that it is drawn from $D_1$. Then we set $\tau_2=1$ so that we can always pick the second value. 
Therefore in expectation we get $5/4$ which is slightly less than $\OPT=6/4$.

To handle the dependencies between the steps, we first distinguish between the events for $k\in [n]$ 
that we pick a value between $\tau_{k+1}$ and $\tau_k$. We show that the expected value 
we pick at such events is indeed highly dependent on $\theta(k)$, the probability of passing the first $k$ elements. 
We then use this observation to analyze competitive ratio with respect to $\theta(k)$'s and 
the thresholds factors $\alpha_k$'s. We finally show that the competitive ratio is indeed maximized 
by choosing the threshold factors described in Theorem~\ref{thm:n:thresholds}. 
In Section~\ref{sec:2thresholds}, we first prove the theorem for the simple case of $n=2$. 
This enables us to demonstrate our techniques without going into the more complicated 
dependencies for general $n$. We then present the full proof of Theorem~\ref{thm:n:thresholds} 
in Section~\ref{sec:generaln}. We would like to emphasize that our algorithm only needs to know the value of $\OPT$, thus requiring only a weak access to the distributions themselves.

As mentioned before, Bayesian optimal mechanism design problems provide a compelling application of prophet inequalities in economics. In such a Bayesian market, we have a set of $n$ agents with private types sampled from (not necessary identical) known distributions. Upon receiving the reported types, a seller has to allocate resources and charge prices to the agents. The goal is to maximize the seller's revenue in equilibrium. Chawla et al.~\cite{CHMS10} pioneered the study the approximability of a special class of such mechanisms, \textit{sequential posted pricing} (SPM): the seller makes a sequence of take-it-or-leave-it offers to agents, offering an item for a specific price. They show although simple, SPMs approximate the optimal revenue in many different settings. 
Therefore prophet inequalities directly translate to approximation factors for the seller's revenue in these settings through standard machineries. Indeed one can analyze the so-called \textit{virtual values} of winning bids introduced by Roger Myerson \cite{myerson1981optimal}, to prove via prophet inequalities that the expected virtual value obtained by the SPM mechanism approximates an offline optimum that is with respect to the exact types. Chawla et al.~\cite{CHMS10} provide a type of prophet inequality in which one can choose the ordering of agents. They show that under matroid feasibility constraints, one can achieve a competitive ratio of $0.5$ in this model, and no algorithm can achieve a ratio better $0.8$. Kleinberg and Weinberg \cite{DBLP:conf/stoc/KleinbergW12} later improved there result by giving an algorithm with the tight competitive ratio of $0.5$ for an adversarial ordering. Our result can be seen as improving their approximation guarantees to $0.63$ for the case of single-item SPMs when the order of agents are chosen randomly.

On the other hand, from the negative side 
the following theorem shows that no online algorithm  
can achieve a competitive ratio better than $0.75$. 
The proof is given in Section \ref{sec:low:bounds}. 

\begin{theorem}
\label{lem:no:alg:2nd}
For any arbitrary small positive number $\epsilon$, 
there is no online algorithm for the prophet secretary problem with competitive ratio $0.75 +\epsilon$.
\end{theorem}

We also consider the minimization variants of the prophet inequality problem, the prophet secretary problem and the secretary problem. In the minimization variant, we need to select one element of the input and we aim to minimize the expected value of the selected element. 
In particular, we show that, even for the simple case in which 
numbers are drawn from \emph{identical and independent distributions (i.i.d.)}, 
there is no $\frac {(1.11)^n} 6$ competitive online algorithm for the minimization variants of the prophet inequality and prophet secretary problems. 

\begin{theorem}
\label{lem:no:alg:3rd}
The competitive ratio of any online algorithm for the minimization prophet inequality 
with $n$ identical and independent distributions is bounded by $\frac {(1.11)^n} 6$. This bound holds for the minimization prophet secretary problem as well.
\end{theorem}

Furthermore, we empower the online algorithm and assume that the online algorithm can withdraw and change its decision once. Indeed, for the minimization variants of all, prophet secretary problem, secretary problem and prophet inequality, we show that there is no $C$ competitive algorithm, where $C$ is an arbitrary large number. 
The proof of Theorem \ref{lem:no:alg:3rd} and this last result are given in Section \ref{sec:low:min}.


\section{Preliminaries}
We first define some notation. 
For every $k\in [n]$, let $z_k$ denote the random variable that shows the value we pick at the $k^{th}$ step. 
Observe that for a fixed sequence of drawn values and a fixed permutation, at most one of $z_k$'s is non-zero 
since we only pick one number. Let $z$ denote the value chosen by the algorithm. 
By definition, $z=\sum_{k=1}^n z_k$. In fact, since all but one of $z_k$'s are zero, we have the following proposition. 
We note that since the thresholds are deterministic, the randomness comes from the permutation $\pi$ and the distributions.

\begin{proposition}\label{prop:zbreak}
$\prob{z\geq x}= \sum_{k\in [n]} \prob{z_k \geq x}$.
\end{proposition}

For every $k\in [n]$, let $\theta(k)$ denote the probability that Algorithm {\sf Prophet Secretary} 
does not choose a value from the first $k$ steps. For every $i\in [n]$ and $k\in [n-1]$, let $q_{-i}(k)$ denote 
the probability that event (i) happens conditioned on event (ii).
\begin{enumerate}[(i)]
\item Algorithm {\sf Prophet Secretary} does not choose a value from the first $k$ elements. 
\item None of the first $k$ values are drawn from $D_i$.
\end{enumerate}

\begin{proposition}\label{prop:thetaq}
If the thresholds of Algorithm {\sf Prophet Secretary} are non-increasing, then
for every $i\in [n]$ and $k\in [n-1]$, we have $\theta(k+1)\leq q_{-i}(k)$.
\end{proposition}

\begin{proof}
In what follows let $i\in [n]$ be a fixed value. The claim is in fact very intuitive: $q_{-i}(k)$ 
is the probability of the event that the algorithm passes $k$ values chosen from all distributions but $D_{i}$. On the other hand,
$\theta(k+1)$ corresponds to the event that the algorithm passes $k+1$ values chosen from all distributions. 
Intuitively, in the latter we have surely passed $k$ values chosen from all but $D_i$. 
Therefore $\theta(k+1)$ cannot be more than $q_{-i}(k)$.

Formalizing the intuition above, however, requires an exact formulation of the probabilities. 
For a permutation $s$ of size $k$ of $[n]$, let $s(j)$, for $j\in [k]$, denote the number at position $j$ of $s$.
For $k\in [n]$, let $S(k)$ denote the set of permutations of size $k$ of $[n]$. Let $S_{-i}(k)$ 
denote the set of permutations of size $k$ of $[n]\backslash \{i\}$.
Observe that
\begin{align*}
\left|S(k)\right| = \frac{n!}{(n-k)!} && \text{ and } &&\left|S_{-i}(k)\right| = \frac{(n-1)!}{(n-1-k)!}
\end{align*}
In particular, we note that $|S(k+1)|=n |S_{-i}(k)|$. We can now write down the exact formula for $q_{-i}(k)$ and $\theta(k+1)$.
\begin{align}
\theta(k+1)=\frac{1}{|S(k+1)|} \sum_{s\in S(k+1)} \prod_{j\in [k+1]} \prob{X_{s(j)}< \tau_j} \label{eq:thetak} \\
q_{-i}(k)=\frac{1}{|S_{-i}(k)|} \sum_{s\in S_{-i}(k)} \prod_{j\in [k]} \prob{X_{s(j)}< \tau_j} \label{eq:qk}
\end{align}

We now show that $\theta(k+1)$ can be written down as a convex combination of terms that are less than $q_{-i}(k)$. 
For every $\ell\in [k]$, let $S(k+1,\ell)=\{s\in S(k+1) | s(\ell)=i \}$. We have 

\[
\begin{split}
  \sum_{s: S(k+1,\ell)} \prod_{j\in [k+1]} \prob{X_{s(j)}\leq \tau_j} 
  & \leq \sum_{s: S(k+1,\ell)} \prod_{j\in [k+1]\backslash \{\ell\}} \prob{X_{s(j)}\leq \tau_j}\\ 
  & \leq \sum_{s: S(k+1,\ell)} \prod_{j\in [k]\backslash \{\ell\}} \prob{X_{s(j)}\leq \tau_j} \times \prob{X_{s(k+1)}\leq \tau_{\ell}}\\
	&= \sum_{s: S_{-i}(k)} \prod_{j\in [k]} \prob{X_{s(j)}\leq \tau_j}  \enspace , 
\end{split}
\]
where the first inequality is since $\prob{X_i<\tau_{\ell}}\leq 1$ and 
the second inequality is because $\prob{X_{s(k+1)}\leq \tau_{k+1}}\leq \prob{X_{s(k+1)}\leq \tau_{\ell}}$ 
and since $\tau_{k+1}\leq \tau_{\ell}$.   

Thus, using Equation~\ref{eq:qk} we have

\begin{align}
\sum_{s: S(k+1,\ell)} \prod_{j\in [k+1]} \prob{X_{s(j)}\leq \tau_j} = \sum_{s: S_{-i}(k)} \prod_{j\in [k]} \prob{X_{s(j)}\leq \tau_j}   
=|S_{-i}(k)| q_{-i}(k) \enspace .
\label{eq:qconvert}
\end{align}

Equation~\ref{eq:qconvert} establishes the relation to $q_{-i}(k)$ for members of $S(k+1)$ 
that contain $i$ in one of the first $k$ positions. 

Let $\overline{S(k+1)}=S(k+1)\backslash \bigcup_{\ell\in [k]} S(k+1,\ell)$. 
We then have 

\[
\begin{split}
   \sum_{s: \overline{S(k+1)}} \prod_{j\in [k+1]\backslash \{\ell\}} \prob{X_{s(j)}\leq \tau_j} 
	& =\sum_{s: \overline{S(k+1)}} \prod_{j\in [k]\backslash \{\ell\}} \prob{X_{s(j)}\leq \tau_j} \times \prob{X_{s(k+1)}\leq \tau_{k+1}}\\
	& \le (n-k) \sum_{s: S_{-i}(k+1)} \prod_{j\in [k]\backslash \{\ell\}} \prob{X_{s(j)}\leq \tau_j}\\
	& =(n-k) |S_{-i}(k)| q_{-i}(k) \enspace,
\end{split}
\]
where the inequality is since $\prob{X_{s(k+1)}\leq \tau_{k+1}}\leq 1 $. 

Finally using this last equation and Equations \ref{eq:thetak} and \ref{eq:qconvert} 
and  since $|S(k+1)|=n  |S_{-i}(k)|$ we obtain 

\[
\begin{split}
  \theta(k+1) &=\quad \frac{1}{|S(k+1)|} \sum_{s\in S(k+1)} \prod_{j\in [k+1]} \prob{X_{s(j)}< \tau_j}\\
	& \frac{1}{|S(k+1)|}\left( \sum_{\ell\in [k]} \sum_{s\in S(k+1,\ell)} \prod_{j\in [k+1]} \prob{X_{s(j)}< \tau_j}  
	+ \sum_{s\in \overline{S(k+1)}} \prod_{j\in [k+1]} \prob{X_{s(j)}< \tau_j} \right) \\
	&\leq \frac{1}{|S(k+1)|}\left(k |S_{-i}(k)| q_{-i}(k) + (n-k) |S_{-i}(k)| q_{-i}(k) \right)\\
	&= q_{-i}(k) \enspace .
\end{split}
\]

\end{proof}


\section{One Threshold Cannot Break $\frac{1}{2}$ Barrier for Prophet Secretary}
\label{sec:1:threshold}
To illustrate that considering at least $2$ thresholds is necessary to beat $\frac{1}{2}$ barrier 
for the prophet secretary problem, we first give an example that shows achieving better than 
$\frac{1}{2}$-competitive ratio for any online algorithm that uses only one threshold for 
the prophet secretary problem is not possible. 
 
\begin{theorem}
\label{thm:1:threshold:no}
There is no online algorithm for the prophet secretary problem that uses one threshold and can achieve 
competitive ratio better than $0.5 + \frac{1}{2n}$. 
\end{theorem}

\begin{proof}
Suppose we have $n+1$ distributions where the first $n$ distributions always gives $\frac 1 {1-1/n}$ 
and the $(n+1)^{th}$ distribution gives $n$ with probability $\frac 1 n$ and gives $0$ with probability $1-\frac 1 n$. 
Therefore, with probability $\frac 1 n$, the maximum is $n$, and 
with probability $1-\frac 1 n$, the maximum is $\frac 1 {1-1/n}$. 
Thus, the expected outcome of the offline optimum algorithm is $\frac 1 n \times n + (1-\frac 1 n)\times \frac 1 {1-1/n} = 2$.

Now, suppose we have an online algorithm that uses one threshold, say $T$ for a number $T$, that is the online algorithm 
accepts the first number greater or equal to a threshold  $T$. 
We consider two cases for $T$.  
The first case is if $T> \frac 1 {1-1/n}$ for which the algorithm does not accept $\frac 1 {1-1/n}$ and thus, 
the expected outcome of such an algorithm is $\frac 1 n \times n = 1$.  

The second case is if $T \leq \frac 1 {1-1/n}$. 
Observe that, with probability $\frac n {n+1}$, the first number is $\frac 1 {1-1/n}$ and the online algorithm accepts it. 
And, with probability $\frac 1 {n+1}$, the distribution that gives $n$ with probability $\frac 1 n$ will be the first 
and the outcome of the algorithm is $1$. 
Thus, the expected outcome of the online algorithm is 
$$\frac n {n+1} \times \frac 1 {1-1/n} + \frac 1 {n+1} \times 1 = \frac {n^2} {n^2-1} + \frac 1 {n+1} \leq 1 + \frac 1 n \enspace .$$

Therefore, the competitive ratio of the online algorithms that uses only one threshold is bounded by 
$\frac {1 + 1/n} 2 = 0.5 + \frac{1}{2n}$.
\end{proof}


\section{Two Thresholds Breaks $\frac{1}{2}$ Barrier} 
\label{sec:2thresholds}
Since using one threshold is hopeless, we now try using two thresholds. 
More formally, for the first half of steps, we use a certain threshold, and then 
we use a different threshold for the rest of steps. We note that similar to the one-threshold algorithm, 
both thresholds should be proportional to $\OPT$.
Furthermore, at the beginning we should be optimistic and try to have a higher threshold, 
but if we cannot pick a value in the first half, we may need to lower the bar! 
We show that by using two thresholds one can indeed achieve the competitive ratio of 
$\frac{5}{9}\simeq 0.55$. In fact, this improvement beyond $1/2$ happens even at $n=2$. 
Thus as a warm up before analyzing the main algorithm with $n$ thresholds, we focus on the case of $n=2$.

Let $\tau_1=\alpha_1 \OPT$ and $\tau_2=\alpha_2 \OPT$ for some $1\geq \alpha_1 \geq \alpha_2 \geq 0$ 
to be optimized later. Recall that $z_1$ and $z_2$ are the random variables showing the values picked up 
by the algorithm at step one and two, respectively. We are interested in comparing $\ex{z}$ with $\OPT$. 
By Proposition~\ref{prop:zbreak} we have
\begin{align*}
\ex{z} &= \int_{0}^{\infty} \prob{z\geq x} dx
 = \int_{0}^{\infty} \prob{z_1\geq x} dx + \int_{0}^{\infty} \prob{z_2\geq x} dx \enspace .
\end{align*}

Observe that $z_1$ (resp. $z_2$) is either zero or has a value more than $\tau_1$ (resp. $\tau_2$). 
In fact, since $\tau_1\geq \tau_2$, $z$ is either zero or has a value more than $\tau_2$. 
Recall that $\theta(1)$ is the probability of $z_1=0$ while $\theta(2)$ is the probability of $z_1=z_2=0$. 
This observation leads to the following simplification:
\begin{align*}
\ex{z} &= \int_{0}^{\tau_2} \prob{z_1\geq x} dx
+\int_{\tau_2}^{\tau_1} \prob{z_1\geq x} dx
+ \int_{\tau_1}^{\infty} \prob{z_1\geq x} dx 
+ \int_{0}^{\tau_2} \prob{z_2\geq x} dx
+\int_{\tau_2}^{\infty} \prob{z_2\geq x} dx\\
&= \int_{0}^{\tau_2} \prob{z\geq x} dx
+\int_{\tau_2}^{\tau_1} \prob{z_1\geq x} dx
\quad + \int_{\tau_1}^{\infty} \prob{z_1\geq x} dx 
+\int_{\tau_2}^{\infty} \prob{z_2\geq x} dx\\
	&= \tau_2(1-\theta(2)) + (\tau_1-\tau_2) (1-\theta(1)) + \int_{\tau_1}^{\infty} \prob{z_1\geq x} dx + \int_{\tau_2}^{\infty} \prob{z_2\geq x} dx \enspace .
\end{align*}

Let us first focus on $\prob{z_1\geq x}$. The first value may come from any of the two distributions, thus we have
\begin{align*}
\prob{z_1\geq x}=\frac{1}{2}\prob{X_1\geq x}+\frac{1}{2}\prob{X_2\geq x} \enspace .
\end{align*}

On the other hand, $z_2$ is non-zero only if we do not pick anything at the first step. 
For $i\in \{1,2\}$, we pick a value of at least $x$ drawn from $D_i$ at step two, if and only if:
(i) the value drawn from $D_i$ is at least $x$; and (ii) our algorithm does not pick a value 
from the previous step which is drawn from the other distribution. By definitions, 
the former happens with probability $\prob{X_i\geq x}$, while the latter happens with probability $q_{-i}(1)$. 
Since these two events are independent we have
\begin{align*}
\prob{z_2\geq x}=\frac{1}{2}\sum_{i\in \{1,2\}} q_{-i}(1) \prob{X_i\geq x} \geq \frac{\theta(2)}{2} \sum_i \prob{X_i\geq x} \enspace .
\end{align*}
where the last inequality follows from Proposition~\ref{prop:thetaq}, although the proposition is trivial for $n=2$. 
We can now continue analyzing $\ex{z}$ from before

\begin{align*}
\ex{z} &= \tau_2(1-\theta(2)) + (\tau_1-\tau_2) (1-\theta(1)) + \int_{\tau_1}^{\infty} \prob{z_1\geq x} dx + \int_{\tau_2}^{\infty} \prob{z_2\geq x} dx \\
&\geq \tau_2(1-\theta(2)) + (\tau_1-\tau_2) (1-\theta(1)) 
+ \frac{\theta(1)}{2} \int_{\tau_1}^{\infty} \sum_i \prob{X_i\geq x} dx + \frac{\theta(2)}{2} \int_{\tau_2}^{\infty} \sum_i \prob{X_i\geq x} dx \enspace . 
\end{align*}

We note that although the $\theta(1)$ factor is not required in the third term of the last inequality, 
we include it so that the formulas can have the same formation as in the general formula of the next sections.

It remains to bound $\int_{\tau_k}^{\infty} \sum_i \prob{X_i\geq x}$ for $k\in \{1,2\}$. 
Recall that $\OPT=\ex{\max_i X_i}$. 
Hence for every $k\in \{1,2\}$ and since $\prob{\max X_i \geq x}\leq 1$ we have 
\[
\OPT =\int_{0}^{\tau_k} \prob{\max X_i \geq x} dx + \int_{\tau_k}^{\infty} \prob{\max X_i \geq x}dx
	\leq \tau_k + \int_{\tau_k}^{\infty} \prob{\max X_i \geq x}dx \enspace .
\]

From $\tau_k=\alpha_k \OPT $ and since $\prob{\max X_i \geq x}\leq \sum_i \prob{X_i\geq x} dx$ we then have 

\[
(1-\alpha_k) \OPT \leq \int_{\tau_k}^{\infty} \prob{\max X_i \geq x}dx	
\leq \int_{\tau_k}^{\infty} \sum_i \prob{X_i\geq x} dx		\enspace .
\]

Therefore we get

\[
  \begin{split}
     \ex{z} &\geq \tau_2(1-\theta(2)) + (\tau_1-\tau_2) (1-\theta(1)) 
		+ \frac{\theta(1)}{2} \int_{\tau_1}^{\infty} \sum_i \prob{X_i\geq x} dx + \frac{\theta(2)}{2} \int_{\tau_2}^{\infty} \sum_i \prob{X_i\geq x} dx \\
&\geq (\alpha_2 \OPT)(1-\theta(2)) + (\alpha_1-\alpha_2) \OPT (1-\theta(1))
+ \frac{\theta(1)}{2} (1-\alpha_1) \OPT+ \frac{\theta(2)}{2} (1-\alpha_2) \OPT\\ 
&=\OPT \left( \alpha_1 + \theta(1) (\frac{1+2\alpha_2-3\alpha_1}{2}) + \theta(2) (\frac{1-3\alpha_2}{2})  \right) \enspace .
 \end{split}
\]

Therefore by choosing $\alpha_2=1/3$ and $\alpha_1=5/9$, 
the coefficients of $\theta(1)$ and $\theta(2)$ become zero, leading to the competitive ratio of $5/9\simeq 0.55$. 
In the next section, we show how one can generalize the arguments to the case of $n$ thresholds for arbitrary $n$.


\section{$(1-\frac{1}{e}\approx 0.63)$-Competitive Ratio Using $n$ Thresholds} \label{sec:generaln}
In this section we prove our main theorem. 
In particular, we invoke Algorithm {\sf Prophet Secretary} with $n$ distinct thresholds $\tau_1,\ldots, \tau_n$. 
The thresholds $\tau_1,\ldots, \tau_n$ that we consider are \textit{non-adaptive} (i.e., 
Algorithm {\sf Prophet Secretary} is oblivious to the history) and \textit{non-increasing}. 
Intuitively, this is because as we move to the end of stream we should be more pessimistic 
and use lower thresholds to catch remaining higher values. 

Formally, for every $k\in [n]$, we consider threshold $\tau_k=\alpha_k \cdot OPT$ where the sequence 
$\alpha_1,\ldots, \alpha_n$ is non-increasing that is, $\alpha_1\ge \alpha_1\ge \ldots \ge \alpha_n$. 
We invoke Algorithm {\sf Prophet Secretary} with these thresholds and analyze the competitive ratio of 
Algorithm {\sf Prophet Secretary} with respect to coefficients $\alpha_k$. 
Theorem~\ref{thm:n:thresholds} shows that there exists a sequence of coefficients $\alpha_k$ 
that leads to the competitive ratio of $(1-1/e)\approx 0.63$.

\vspace{0.2cm}
\noindent
\textbf{Proof of Theorem~\ref{thm:n:thresholds}.}
We prove the theorem in two steps: 
First, we find a lower bound on $\ex{z}$ in terms of $OPT$ and coefficients $\alpha_i$. 
Second, we set coefficients $\alpha_k$ so that $(1)$ $\alpha_1$ becomes the competitive ratio of Algorithm {\sf Prophet Secretary} 
and $(2)$ $\alpha_1$ converges to $1-1/e$, when $n$ goes to infinity. 

We start by proving the following auxiliary lemmas. 
In the first lemma, we find a lower bound for $\int_{\tau_k}^{\infty} \prob{\max X_i\geq x} dx$ based on 
$\OPT=\ex{\max_i X_i}$.

\begin{lemma}
\label{lem:lower:max:prob}
$ \int_{\tau_k}^{\infty} \prob{\max X_i\geq x} dx \ge (1-\alpha_k) \OPT$. 
\end{lemma}

\begin{proof}
For an arbitrary $k\in [n]$ and since $\prob{\max X_i \geq x}\leq 1$, we have
\begin{align*}
\OPT &= \ex{\max_i X_i} = \int_{0}^{\infty} \prob{\max X_i\geq x} dx 
\\
&
= \int_{0}^{\tau_k} \prob{\max X_i\geq x} dx + \int_{\tau_k}^{\infty} \prob{\max X_i\geq x} dx \\
&\leq \tau_k + \int_{\tau_k}^{\infty} \prob{\max X_i\geq x} dx \enspace .
\end{align*}

Since, by definition $\tau_k=\alpha_k \cdot\OPT$, we have 
$(1-\alpha_k) \cdot \OPT \leq \int_{\tau_k}^{\infty} \prob{\max X_i\geq x} dx$.
\end{proof}

Next, to find a lower bound for $\ex{z}$, we first split it into two terms.
Later, we find lower bounds for each one of these 
terms based on $\OPT=\ex{\max_i X_i}$. 

\begin{lemma}
\label{lem:two:parts}
Let $z=\sum_{k=1}^n z_k$ denote the value chosen by Algorithm {\sf Prophet Secretary}. 
For $z$ we have 
$$\ex{z}=\sum_{k=1}^n \int_0^{\tau_k} \prob{z_k\geq x} dx + \sum_{k=1}^n \int_{\tau_k}^{\infty} \prob{z_k\geq x} dx \enspace .$$ 
\end{lemma}

\begin{proof}
By Proposition~\ref{prop:zbreak} we have
\begin{align}
\ex{z} &= \int_{0}^{\infty} \prob{z\geq x} dx = \sum_{k=1}^{n} \int_{0}^{\infty} \prob{z_k\geq x} dx \notag \\
&= \sum_{k=1}^n \int_0^{\tau_k} \prob{z_k\geq x} dx + \sum_{k=1}^n \int_{\tau_k}^{\infty} \prob{z_k\geq x} dx \notag\enspace ,
\end{align}
where we use this fact that $z=\sum_{k=1}^n z_k$ because 
we only pick one number for a fixed sequence of drawn values and a fixed permutation and therefore, 
at most one of $z_k$'s is non-zero. 
\end{proof}



\begin{lemma}
\label{lem:first:part}
$\sum_{k=1}^n \int_0^{\tau_k} \prob{z_k\geq x} dx \ge \OPT \sum_{k=1}^{n} (1-\theta(k))(\alpha_k -\alpha_{k+1})$.
\end{lemma}

\begin{proof}
Suppose $x\leq \tau_k$. 
Observe that the event $z_k\geq x$ occurs when Algorithm {\sf Prophet Secretary} chooses a value at step $k$. 
In fact, since the thresholds are non-increasing, whatever we pick at the first $k$ steps would be at least $x$. 
Recall that for every $k\in [n]$, $\theta(k)$ is the probability that Algorithm {\sf Prophet Secretary} 
does not choose a value from the first $k$ steps. 
Hence, for every $k\in[n]$ and $x\leq \tau_k$ we have 
\begin{align}
\sum_{j\leq k} \prob{z_j\geq x}=1-\theta(k) \enspace . \label{eq:belowT}
\end{align}

To simplify the notation, we assume that $\alpha_0=\infty$ which means $\tau_0=\infty$ and 
we let $\alpha_{n+1}=0$ which means $\tau_{n+1}=0$. Therefore we have
$$ \sum_{k=1}^n \int_0^{\tau_k} \prob{z_k\geq x} dx  = \sum_{k=1}^{n} \int_{\tau_{n+1}}^{\tau_k} \prob{z_k\geq x} dx \enspace . $$

Next, we use Equation~(\ref{eq:belowT}) to prove the lemma as follows. 
\begin{align*}
\sum_{k=1}^n \int_0^{\tau_k} \prob{z_k\geq x} dx  
&= \sum_{k=1}^{n} \int_{\tau_{n+1}}^{\tau_k} \prob{z_k\geq x} dx \\
&= \sum_{k=1}^{n} \sum_{r=k}^{n} \int_{\tau_{r+1}}^{\tau_{r}} \prob{z_k\geq x} dx \\
&=\sum_{r=1}^{n} \int_{\tau_{r+1}}^{\tau_{r}} \sum_{k=1}^{r} \prob{z_k\geq x} dx
\geq \sum_{r=1}^{n} \int_{\tau_{r+1}}^{\tau_{r}} (1-\theta(r)) dx \\
&=\sum_{r=1}^{n} (1-\theta(r))(\tau_r -\tau_{r+1})
= \OPT \cdot\sum_{k=1}^{n} (1-\theta(k))(\alpha_k -\alpha_{k+1}) \enspace .
\end{align*}
\end{proof}


\begin{lemma}
\label{lem:second:part}
$\sum_{k=1}^n \int_{\tau_k}^{\infty} \prob{z_k\geq x} dx \ge \OPT \sum_k \frac{\theta(k)}{n} (1-\alpha_k)$.
\end{lemma}
\begin{proof}
Recall that for every distribution $D_i$ we draw a number $X_i$. 
Later, we randomly permute the numbers $X_1,\cdots,X_n$. 
Let the sequence of indices after the random permutation be $\pi_1,\ldots, \pi_n$ that is, 
at step $k$,  number $X_{\pi_k}$ for $\pi_k\in [n]$ is revealed.  

Suppose $x\geq \tau_k$. 
We break the event $z_k>0$ to $n$ different scenarios depending on which index of 
the distributions $D_1,\cdots,D_n$ is mapped to  index $\pi_k$ in the random permutation. 
Let us consider the scenario in which Algorithm {\sf Prophet Secretary} chooses the value 
drawn from a distribution $i$ at step $k$. 
Such a scenario happens if (i) Algorithm {\sf Prophet Secretary} does not choose a value from 
the first $k-1$ steps which are not drawn from $i$, and (ii) $X_i\geq \tau_k$.  
Observe that the two events are independent. Therefore, we have 
$\prob{z_k\geq x} = \sum_i \prob{\pi_k=i}\cdot  \prob{X_i\geq x} \cdot q_{-i}(k-1)$.
Since $\pi_k$ is an index in the random permutation we obtain 
$$
\prob{z_k\geq x} = \sum_i \prob{\pi_k=i}  \prob{X_i\geq x} \cdot q_{-i}(k-1) 
= \frac{1}{n}  \sum_i \prob{X_i\geq x}  q_{-i}(k-1) \enspace .
$$
Using Proposition \ref{prop:thetaq} and an application of the union bound we then have 
\begin{align}
\prob{z_k\geq x} &= \sum_i \prob{\pi_k=i} \cdot \prob{X_i\geq x} \cdot q_{-i}(k-1) \notag\\
	&= \frac{1}{n} \sum_i \prob{X_i\geq x} \cdot q_{-i}(k-1)  \notag\\
	&\geq \frac{\theta(k)}{n} \cdot \sum_i \prob{X_i\geq x} \notag
	\geq \frac{\theta(k)}{n} \cdot \prob{\max_i X_i \geq x} \enspace . \notag 
\end{align}
Therefore, we obtain the following lower bound on $\sum_{k=1}^n \int_{\tau_k}^{\infty} \prob{z_k\geq x} dx$. 
\begin{align*}
\sum_{k=1}^n \int_{\tau_k}^{\infty} \prob{z_k\geq x} dx 
&\geq \sum_k \int_{\tau_k}^{\infty} \frac{\theta(k)}{n} \prob{\max X_i\geq x} dx\\
&=\sum_k \frac{\theta(k)}{n} \int_{\tau_k}^{\infty} \prob{\max X_i\geq x} dx \enspace . 
\end{align*}
Finally, we use the lower bound of Lemma \ref{lem:lower:max:prob} for $\int_{\tau_k}^{\infty} \prob{\max X_i\geq x} dx$ 
to prove the lemma. 
\begin{align*}
\sum_{k=1}^n \int_{\tau_k}^{\infty} \prob{z_k\geq x} dx 
&\ge \sum_k \frac{\theta(k)}{n} \int_{\tau_k}^{\infty} \prob{\max X_i\geq x} dx\\
&\geq \sum_k \frac{\theta(k)}{n} \cdot (1-\alpha_k) \cdot \OPT
= \OPT \cdot \sum_k \frac{\theta(k)}{n} \cdot (1-\alpha_k) \enspace .
\end{align*}

\end{proof}


Now we can plug in the lower bounds of Lemmas \ref{lem:first:part} and \ref{lem:second:part} 
into Lemma \ref{lem:two:parts} to obtain a lower bound for 
$\ex{z}$.
\begin{corollary}
\label{lem:alltogether}
Let $z=\sum_{k=1}^n z_k$ denote the value chosen by Algorithm {\sf Prophet Secretary}. 
For $z$ we have 
$$\ex{z}\geq \OPT \cdot( \alpha_1 + \sum_{k=1}^n \theta(k) (\frac{1}{n}-\frac{\alpha_k}{n}-\alpha_k+\alpha_{k+1})   ) \enspace.$$ 
\end{corollary}

%

We finish the proof of the theorem by proving the following claim.


\begin{lemma}
\label{lem:alpha:1}
The competitive ratio of Algorithm {\sf Prophet Secretary} is at least $\alpha_1$ which 
quickly converges to $1-1/e\approx 0.63$ when $n$ goes to infinity. 
\end{lemma}

\begin{proof}
Using Corollary~\ref{lem:alltogether}, for $z$ we have
$$\ex{z}\geq \OPT  \left( \alpha_1 + \sum_{k=1}^n \theta(k) \left(\frac{1}{n}-\frac{\alpha_k}{n}-\alpha_k+\alpha_{k+1}\right)   \right) \enspace ,$$
which means that the competitive ratio depends on the probabilities $\theta(k)$'s. 
However, we can easily get rid of the probabilities $\theta(k)$'s by
choosing $\alpha_k$'s such that for every $k$, $\left(\frac{1}{n}-\frac{\alpha_k}{n}-\alpha_k+\alpha_{k+1}\right)=0$.

More formally, by starting from $\alpha_{n+1}=0$ and choosing $\alpha_k=\frac{1+n\alpha_{k+1}}{1+n}$ for $k\leq n$, 
the competitive ratio of the algorithm would be $\alpha_1$. Below, we show that when $n$ goes to infinity, $\alpha_1$ 
quickly goes to $1-1/e$ which means that the competitive ratio of 
Algorithm {\sf Prophet Secretary} converges to $1-1/e\approx 0.63$. 

First, we show by the induction that $\alpha_k=\sum_{i=0}^{n-k} \frac{n^{i}}{(1+n)^{i+1}}$. 
For the base case we have 

\begin{align*}
\alpha_n=\frac{1+n\alpha_{n+1}}{1+n}=\frac{1+n\times 0}{1+n} = \frac{n^0}{(1+n)^1} \enspace .
\end{align*}

Given $\alpha_{k+1}=\sum_{i=0}^{n-(k+1)} \frac{n^{i}}{(1+n)^{i+1}}$ we show the equality for $\alpha_k$ as follows.

\begin{align*}
\alpha_k 
 =\frac{1+n\alpha_{k+1}}{1+n}
 &= \frac{1+n(\sum_{i=0}^{n-(k+1)} \frac{n^i}{(1+n)^{i+1}})}{1+n}\\
 &=\frac{n^0}{(1+n)^1} + \sum_{i=0}^{n-(k+1)} \frac{n^{i+1}}{(1+n)^{i+2}}\\
 &=\sum_{i=0}^{n-k} \frac{n^{i}}{(1+n)^{i+1}} \enspace .
\end{align*}

Now we are ready to show $\alpha_1\geq  1-1/e$ when $n$ goes to infinity.
\begin{align*}
 \lim_{n\rightarrow \infty} \alpha_1 &=  \lim_{n\rightarrow \infty} \sum_{i=0}^{n-1} \frac {n^i} {(n+1)^{i+1}}
 =  \lim_{n\rightarrow \infty} \frac 1 {n+1} \sum_{i=0}^{n-1} (1-\frac 1{n+1})^i \\
 & \approx  \lim_{n\rightarrow \infty} \frac 1{n+1} \sum_{i=0}^{n-1} e^{-i/n}
   \approx \int_0^1  e^{-x} dx= 1-1/e\enspace. 
\end{align*}
\end{proof}

\section{Lower Bounds for Prophet Secretary and Minimization Variants of Classical Stopping Theory Problems}
\label{sec:low:bounds}
In this section we give our lower bounds for the prophet secretary problem and the minimization variants of 
the prophet inequality and prophet secretary problems. 
First, in Section \ref{sec:low:prop:sec} we show that no online algorithm for the prophet secretary problem 
can achieve a competitive ratio better than $0.75$. 
Later, in Section \ref{sec:low:min} we consider the minimization variant of the prophet inequality problem. 
We show that, even for the simple case in which numbers are drawn from \emph{identical and independent distributions (i.i.d.)}, 
there is no constant competitive online algorithm for the minimization variants of the prophet inequality and prophet 
secretary problems. 

\subsection{$0.75$-Lower Bound of Prophet Secretary Problem}
\label{sec:low:prop:sec}


Now we prove Theorem~\ref{lem:no:alg:2nd} that improves the above lower bound by showing 
that there is no algorithm for the prophet secretary 
problem with competitive ratio $0.75 +\epsilon$.

\paragraph{Proof of Theorem~\ref{lem:no:alg:2nd}.}
It is known that no algorithm can guarantee a competitive ratio better than $0.5+\epsilon$, 
where $\epsilon$ is an arbitrary small positive number  less than $1$. 
A hard example that shows this upper bound is as follow. 
We have two distributions. The first distribution always gives $1$, and the second distribution gives 
$\frac 1 \epsilon$ with probability $\epsilon$ and $0$ with probability $1-\epsilon$. 
Observe that  if we either accept the first number or reject it our expected 
outcome is $1$. On the other hand, the offline optimum algorithm takes $\frac 1 \epsilon$ with probability 
$\epsilon$ and $1$ with probability $1-\epsilon$. Therefore, the expected outcome of the offline optimum algorithm is 
$\epsilon\cdot\frac 1 \epsilon+ 1 \cdot(1-\epsilon) = 2-\epsilon$, 
and the competitive ratio is at most $\frac 1 {2-\epsilon} \leq 0.5 + \epsilon$.

The above example contains exactly two distributions. 
Thus, if the drawn numbers from these distributions arrive in a random order, 
with probability $0.5$ the arrival order is the worst case order. 
This means that in the prophet secretary, with probability $0.5$ the expected 
outcome of any algorithm on this example is at most $1$, while the offline optimum algorithm is always $2-\epsilon$. 
Therefore, there is no algorithm for the prophet secretary problem with competitive ratio better than 
$\frac {0.5\times 1+0.5\times(2-\epsilon)}{2-\epsilon}\leq 0.75+\epsilon$.


\subsection{Lower Bounds for Minimization Variants of Classical Stopping Theory Problems}
\label{sec:low:min}
In this section, we consider the minimization variant of the prophet secretary problem, in which we need to select one element from the input and we aim to minimize the value of the selected element. Indeed, we provide some hardness results for very simple and restricted cases of this problem. We extend our results to the minimization variant of the classical prophet inequality or minimization variant of secretary problem.

First, we consider the simple case that numbers are drawn from identical and independent distributions. 
In particular, we prove Theorem \ref{lem:no:alg:3rd} that shows, 
even for the simple case of identical and independent distributions, 
there is no $\frac {(1.11)^n}6$ competitive online algorithm for 
the minimization variant of the prophet inequality problem. 
Since the input items come from identical distributions, independently, randomly reordering them do not change the distribution of the items in the input. Thus, this result holds for the minimization variant of the prophet secretary problem as well.

Later, we give more power to the online algorithm and let it to change its decision once; 
we call this model \emph{online algorithm with one exchange}. 
In Theorem \ref{thm:min+1} we show that, for any large number $C$  
there is no $C$ competitive online algorithm 
with one exchange for the minimization variant of the prophet inequality. In Corollary \ref{cr:MinProphSec} and Corollary \ref{cr:MinSec} we extend this hardness result to the minimization variant of prophet secretary problem and secretary problem, respectively.


\paragraph{Proof of Theorem~\ref{lem:no:alg:3rd}.}
Suppose we have $n$ identical distributions, each gives $0$ with probability $\frac 1 3$, 
$1$ with probability $\frac 1 3$ and $2^n$ with probability $\frac 1 3$. One can see that 
with probability $(\frac 1 3)^n$, all of the numbers are $2^n$ and thus the minimum number is $2^n$. 
Also, with probability $(\frac 2 3)^n-(\frac 1 3)^n$, there is no $0$ and there is at least one $1$, 
and thus, the minimum is $1$. In all the other cases, the minimum is $0$. 
Therefore, the expected outcome of the offline optimum algorithm is 
$(\frac 1 3)^n\times 2^n + (\frac 2 3)^n-(\frac 1 3)^n < \frac {2^{n+1}} {3^n}$.

For this example, without loss of generality we can assume that any online algorithm accept $0$ 
as soon as it appears, and also it does not accept $2^n$ except for the last item. 
Assume $i+1$ is the first time that the algorithm is willing to accept $1$. 
The probability that we arrive in this point is $(\frac 2 3)^i$ and the probability that 
we see $1$ at that point is $\frac 1 3$. On the other hand, the probability that such an algorithm 
does not accept anything up to the last number and sees that the last number is $2^n$ is 
at least $(\frac 2 3)^i(\frac 1 3)^{n-i}$. 
Therefore, the expected outcome of the online algorithm is at least 
$(\frac 2 3)^i \frac 1 3\times 1+ (\frac 2 3)^i(\frac 1 3)^{n-i}\times 2^n = \frac {2^i3^{n-i-1}+3^i} {3^n}$. 
Thus, the competitive ratio is at least
\begin{align*}
\frac {2^i 3^{n-i-1}+3^i} {2^{n+1}} = \frac 1 6 ( \frac {3^{n-i}}{2^{n-i}} + \frac{3^{i+1}}{2^{n}})\enspace .
\end{align*}
If $i\leq 0.73n$, the left term is at least $\frac {(1.11)^n} 6$; otherwise, the right term is at least $\frac {(1.11)^n}6$.


\begin{theorem}
\label{thm:min+1}
For any large positive number $C$, there is no $C$-competitive algorithm for minimization prophet inequality with one exchange. 
\end{theorem}

\begin{proof}
Suppose we have three distributions as follows. The first distribution always gives $1$. 
The second distribution gives $\frac 1 \epsilon$ with probability $\epsilon$ and 
gives $\frac \epsilon {1-\epsilon}$ with probability $1-\epsilon$. 
The third distribution gives $\frac 1 \epsilon$ with probability $\epsilon$ and gives $0$ with probability $1-\epsilon$. We set $\epsilon$ later.

We observe that the minimum number is $0$ with probability $1-\epsilon$, is $\frac \epsilon {1-\epsilon}$ 
with probability $\epsilon(1-\epsilon)$ and is $1$ with probability $\epsilon^2$. 
Thus, the expected outcome of the optimum algorithm is 
$$(1-\epsilon)\times 0+\epsilon(1-\epsilon)\times \frac \epsilon {1-\epsilon} + \epsilon^2 \times 1 = 2\epsilon^2 \enspace .$$

Now, we show that the outcome of any online algorithm with one exchange for this input is at least $\epsilon$. 
In fact, this means that, the competitive ratio can not be less than $\frac {\epsilon}{2\epsilon^2}=\frac {1}{2\epsilon}$. 
If we set $\epsilon = \frac {1}{2C}$, this means that there is no $C$-competitive algorithm for the minimization prophet inequality 
with one exchange.

Recall that there is no uncertainty in the first number. 
If an algorithm withdraw the first number, it can take the outcome of either the second distribution or the third distribution. 
However, if we do this, with probability ${\epsilon^2}$ the outcome is $\frac 1 \epsilon$. 
Thus, the expected outcome is at least $\epsilon$ as desired.

Now, we just need to show that if an algorithm does not select the first number, 
its expected outcome is at least $\epsilon$. Observe that if this happens, 
then the algorithm has only the option of choosing either the second and the third distribution. 
We consider two cases. The first case is if the algorithm does not select the second number, 
then it must select the third number. Therefore,  with probability $\epsilon$, 
the outcome of the algorithm is $1$ and thus, the expected outcome is at least $\epsilon$. 
The second case occurs if the algorithm selects the second number and that is $\frac{\epsilon}{1-\epsilon}$. 
Therefore, with probability $1-\epsilon$, the outcome of the algorithm is $\frac{\epsilon}{1-\epsilon}$ 
and again the expected outcome is $\frac{\epsilon}{1-\epsilon}\times(1-\epsilon)= \epsilon$.

\end{proof}

Finally, we show that even for the minimization variant of the prophet secretary there is no hope 
to get a constant competitive ratio. 

\begin{corollary}
\label{cr:MinProphSec}
For any large number $C$, there is no $C$-competitive algorithm for minimization prophet secretary with one exchange. 
\end{corollary}

\begin{proof}
In Theorem ~\ref{thm:min+1}, we have three distributions. 
Thus, in the prophet secretary model the worst case order happen with probability $\frac 1 6$. 
Thus, the competitive ratio can not be more than $\frac {\frac 1 6 \cdot\epsilon}{2\epsilon^2}=\frac 1 {12 \epsilon}$. 
If we set $\epsilon = \frac {1}{12C}$, this essentially means that there is no $C$-competitive algorithm for 
the minimization prophet secretary with one exchange.
\end{proof}

\begin{corollary}
		\label{cr:MinSec}
For any large number $C$ there is no $C$-competitive algorithm for the minimization secretary problem 
with one exchange. 
\end{corollary}

\begin{proof}
Suppose for the sake of contradiction that there exists an algorithm $Alg$ for the minimization secretary problem 
which is $C$-competitive. Consider all possible realizations of the example in Theorem ~\ref{thm:min+1}. 
Algorithm $Alg$ is $C$-competitive when each of the these realizations comes in a random order. 
Therefore, Algorithm $Alg$ is $C$-competitive when the input is a distribution over these realizations. 
This says that $Alg$ is a $C$-competitive algorithm for the minimization prophet secretary, 
which contradicts Corollary ~\ref{cr:MinProphSec} and completes the proof. 
\end{proof}


\noindent
\textbf{Acknowledgments.}
We would like to thank Robert Kleinberg  for fruitful discussion on early stages 
of this project.





\newcommand{\Proc}{Proceedings of the~}
\newcommand{\STOC}{Annual ACM Symposium on Theory of Computing (STOC)}
\newcommand{\FOCS}{IEEE Symposium on Foundations of Computer Science (FOCS)}
\newcommand{\SODA}{Annual ACM-SIAM Symposium on Discrete Algorithms (SODA)}
\newcommand{\SOCG}{Annual Symposium on Computational Geometry (SoCG)}
\newcommand{\ICALP}{Annual International Colloquium on Automata, Languages and Programming (ICALP)}
\newcommand{\ESA}{Annual European Symposium on Algorithms (ESA)}
\newcommand{\CCC}{Annual IEEE Conference on Computational Complexity (CCC)}
\newcommand{\RANDOM}{International Workshop on Randomization and Approximation Techniques in Computer Science (RANDOM)}
\newcommand{\PODS}{ACM SIGMOD Symposium on Principles of Database Systems (PODS)}
\newcommand{\SSDBM}{ International Conference on Scientific and Statistical Database Management (SSDBM)}
\newcommand{\ALENEX}{Workshop on Algorithm Engineering and Experiments (ALENEX)}
\newcommand{\BEATCS}{Bulletin of the European Association for Theoretical Computer Science (BEATCS)}
\newcommand{\CCCG}{Canadian Conference on Computational Geometry (CCCG)}
\newcommand{\CIAC}{Italian Conference on Algorithms and Complexity (CIAC)}
\newcommand{\COCOON}{Annual International Computing Combinatorics Conference (COCOON)}
\newcommand{\COLT}{Annual Conference on Learning Theory (COLT)}
\newcommand{\COMPGEOM}{Annual ACM Symposium on Computational Geometry}
\newcommand{\DCGEOM}{Discrete \& Computational Geometry}
\newcommand{\DISC}{International Symposium on Distributed Computing (DISC)}
\newcommand{\ECCC}{Electronic Colloquium on Computational Complexity (ECCC)}
\newcommand{\FSTTCS}{Foundations of Software Technology and Theoretical Computer Science (FSTTCS)}
\newcommand{\ICCCN}{IEEE International Conference on Computer Communications and Networks (ICCCN)}
\newcommand{\ICDCS}{International Conference on Distributed Computing Systems (ICDCS)}
\newcommand{\VLDB}{ International Conference on Very Large Data Bases (VLDB)}
\newcommand{\IJCGA}{International Journal of Computational Geometry and Applications}
\newcommand{\INFOCOM}{IEEE INFOCOM}
\newcommand{\IPCO}{International Integer Programming and Combinatorial Optimization Conference (IPCO)}
\newcommand{\ISAAC}{International Symposium on Algorithms and Computation (ISAAC)}
\newcommand{\ISTCS}{Israel Symposium on Theory of Computing and Systems (ISTCS)}
\newcommand{\JACM}{Journal of the ACM}
\newcommand{\LNCS}{Lecture Notes in Computer Science}
\newcommand{\RSA}{Random Structures and Algorithms}
\newcommand{\SPAA}{Annual ACM Symposium on Parallel Algorithms and Architectures (SPAA)}
\newcommand{\STACS}{Annual Symposium on Theoretical Aspects of Computer Science (STACS)}
\newcommand{\SWAT}{Scandinavian Workshop on Algorithm Theory (SWAT)}
\newcommand{\TALG}{ACM Transactions on Algorithms}
\newcommand{\UAI}{Conference on Uncertainty in Artificial Intelligence (UAI)}
\newcommand{\WADS}{Workshop on Algorithms and Data Structures (WADS)}
\newcommand{\SICOMP}{SIAM Journal on Computing}
\newcommand{\JCSS}{Journal of Computer and System Sciences}
\newcommand{\JASIS}{Journal of the American society for information science}
\newcommand{\PMS}{ Philosophical Magazine Series}
\newcommand{\ML}{Machine Learning}
\newcommand{\DCG}{Discrete and Computational Geometry}
\newcommand{\TODS}{ACM Transactions on Database Systems (TODS)}
\newcommand{\PHREV}{Physical Review E}
\newcommand{\NATS}{National Academy of Sciences}
\newcommand{\MPHy}{Reviews of Modern Physics}
\newcommand{\NRG}{Nature Reviews : Genetics}
\newcommand{\BullAMS}{Bulletin (New Series) of the American Mathematical Society}
\newcommand{\AMSM}{The American Mathematical Monthly}
\newcommand{\JAM}{SIAM Journal on Applied Mathematics}
\newcommand{\JDM}{SIAM Journal of  Discrete Math}
\newcommand{\JASM}{Journal of the American Statistical Association}
\newcommand{\AMS}{Annals of Mathematical Statistics}
\newcommand{\JALG}{Journal of Algorithms}
\newcommand{\TIT}{IEEE Transactions on Information Theory}
\newcommand{\CM}{Contemporary Mathematics}
\newcommand{\JC}{Journal of Complexity}
\newcommand{\TSE}{IEEE Transactions on Software Engineering}
\newcommand{\TNDE}{IEEE Transactions on Knowledge and Data Engineering}
\newcommand{\JIC}{Journal Information and Computation}
\newcommand{\ToC}{Theory of Computing}
\newcommand{\Algorithmica}{Algorithmica}
\newcommand{\MST}{Mathematical Systems Theory}
\newcommand{\Com}{Combinatorica}
\newcommand{\NC}{Neural Computation}
\newcommand{\TAP}{The Annals of Probability}

\bibliographystyle{abbrv}
\bibliography{References}

\begin{thebibliography}{10}

\bibitem{ajtai01}
M.~Ajtai, N.~Megiddo, and O.~Waarts.
\newblock Improved algorithms and analysis for secretary problems and
  generalizations.
\newblock {\em SIAM J. Discrete Math.}, 14(1):1--27, 2001.

\bibitem{A11}
S.~Alaei.
\newblock Bayesian combinatorial auctions: Expanding single buyer mechanisms to
  many buyers.
\newblock In {\em FOCS}. 2011.

\bibitem{DBLP:conf/sigecom/AlaeiHL12}
S.~Alaei, M.~Hajiaghayi, and V.~Liaghat.
\newblock Online prophet-inequality matching with applications to ad
  allocation.
\newblock In {\em EC}, 2012.

\bibitem{alaei2013online}
S.~Alaei, M.~Hajiaghayi, and V.~Liaghat.
\newblock The online stochastic generalized assignment problem.
\newblock In {\em Approx}. Springer, 2013.

\bibitem{alaei2011adcell}
S.~Alaei, M.~T. Hajiaghayi, V.~Liaghat, D.~Pei, and B.~Saha.
\newblock Adcell: Ad allocation in cellular networks.
\newblock In {\em ESA}. 2011.

\bibitem{Assaf01ratioprophet}
D.~Assaf, L.~Goldstein, and E.~Samuel-cahn.
\newblock Ratio prophet inequalities when the mortal has several choices.
\newblock {\em Ann. Appl. Prob}, 12:972--984, 2001.

\bibitem{BIKK07}
M.~Babaioff, N.~Immorlica, D.~Kempe, and R.~Kleinberg.
\newblock A knapsack secretary problem with applications.
\newblock In {\em APPROX}, pages 16--28, 2007.

\bibitem{BIKK08}
M.~Babaioff, N.~Immorlica, D.~Kempe, and R.~Kleinberg.
\newblock Online auctions and generalized secretary problems.
\newblock {\em SIGecom Exch.}, 7(2):1--11, 2008.

\bibitem{BIK07}
M.~Babaioff, N.~Immorlica, and R.~Kleinberg.
\newblock Matroids, secretary problems, and online mechanisms.
\newblock In {\em SODA}, pages 434--443, 2007.

\bibitem{DBLP:journals/talg/BateniHZ13}
M.~Bateni, M.~T. Hajiaghayi, and M.~Zadimoghaddam.
\newblock Submodular secretary problem and extensions.
\newblock {\em ACM Transactions on Algorithms}, 9(4):32, 2013.

\bibitem{CHMS10}
S.~Chawla, J.~Hartline, D.~Malec, and B.~Sivan.
\newblock Multi-parameter mechanism design and sequential posted pricing.
\newblock 2010.

\bibitem{Dynkin63}
E.~B. Dynkin.
\newblock The optimum choice of the instant for stopping a markov process.
\newblock {\em Sov. Math. Dokl.}, 4:627--629, 1963.

\bibitem{F83}
P.~R. Freeman.
\newblock The secretary problem and its extensions: a review.
\newblock {\em Internat. Statist. Rev.}, 51(2):189--206, 1983.

\bibitem{GM66}
J.~P. Gilbert and F.~Mosteller.
\newblock Recognizing the maximum of a sequence.
\newblock {\em J. Amer. Statist. Assoc.}, 61:35--73, 1966.

\bibitem{GHB83}
K.~S. Glasser, R.~Holzsager, and A.~Barron.
\newblock The {$d$} choice secretary problem.
\newblock {\em Comm. Statist. C---Sequential Anal.}, 2(3):177--199, 1983.

\bibitem{HKP04}
M.~T. Hajiaghayi, R.~Kleinberg, and D.~C. Parkes.
\newblock Adaptive limited-supply online auctions.
\newblock In {\em EC}, pages 71--80, 2004.

\bibitem{HKS07}
M.~T. Hajiaghayi, R.~Kleinberg, and T.~Sandholm.
\newblock Automated online mechanism design and prophet inequalities.
\newblock In {\em AAAI}, pages 58--65, 2007.

\bibitem{IKM06}
N.~Immorlica, R.~D. Kleinberg, and M.~Mahdian.
\newblock Secretary problems with competing employers.
\newblock In {\em WINE}, pages 389--400, 2006.

\bibitem{K78}
D.~P. Kennedy.
\newblock Prophet-type inequalities for multi-choice optimal stopping.
\newblock In {\em In Stoch. Proc. Applic.} 1978.

\bibitem{K05}
R.~Kleinberg.
\newblock A multiple-choice secretary algorithm with applications to online
  auctions.
\newblock In {\em SODA}, pages 630--631, 2005.

\bibitem{DBLP:conf/stoc/KleinbergW12}
R.~Kleinberg and S.~M. Weinberg.
\newblock Matroid prophet inequalities.
\newblock In {\em STOC}, 2012.

\bibitem{KS77}
U.~Krengel and L.~Sucheston.
\newblock Semiamarts and finite values.
\newblock In {\em Bull. Am. Math. Soc.} 1977.

\bibitem{KS78}
U.~Krengel and L.~Sucheston.
\newblock On semiamarts, amarts, and processes with finite value.
\newblock In {\em In Kuelbs, J., ed., Probability on Banach Spaces}. 1978.

\bibitem{L14}
O.~Lachish.
\newblock $o(\log\log \text{rank})$-competitive-ratio for the matroid secretary
  problem.
\newblock In {\em FOCS}. 2014.

\bibitem{myerson1981optimal}
R.~B. Myerson.
\newblock Optimal auction design.
\newblock {\em Mathematics of operations research}, 6(1):58--73, 1981.

\bibitem{Vanderbei80}
R.~J. Vanderbei.
\newblock The optimal choice of a subset of a population.
\newblock {\em Math. Oper. Res.}, 5(4):481--486, 1980.

\bibitem{Wilson91}
J.~G. Wilson.
\newblock Optimal choice and assignment of the best {$m$} of {$n$} randomly
  arriving items.
\newblock {\em Stochastic Process. Appl.}, 39(2):325--343, 1991.

\bibitem{DBLP:conf/soda/Yan11}
Q.~Yan.
\newblock Mechanism design via correlation gap.
\newblock In {\em Proceedings of the Twenty-Second Annual {ACM-SIAM} Symposium
  on Discrete Algorithms, {SODA} 2011, San Francisco, California, USA, January
  23-25, 2011}, pages 710--719, 2011.

\end{thebibliography}


\end{document}